\documentclass[12pt]{amsart}

\usepackage{amssymb,amsthm,amsmath}
\usepackage[numbers,sort&compress]{natbib}
\usepackage{color}
\usepackage{graphicx}
\usepackage{tikz}

\title[ ]{Upper bounds on transport exponents for long range  operators}
\author{ Svetlana Jitomirskaya}
\address[ Svetlana Jitomirskaya]{ Department of Mathematics, University of California, Irvine, California 92697-3875, USA}
\email{szhitomi@math.uci.edu}

\author{Wencai Liu}
\address[Wencai Liu]{ Department of Mathematics, University of California, Irvine, California 92697-3875, USA}
\address{ Current address: Department of Mathematics, Texas A\&M University, College Station, TX 77843-3368, USA} \email{liuwencai1226@gmail.com; wencail@tamu.edu}

\newcommand{\R}{\mathbb{R}}
\newcommand{\Z}{\mathbb{Z}}

\newcommand{\C}{\mathbb{C}}
\newcommand{\T}{\mathbb{T}}

\theoremstyle{plain}
\newtheorem{theorem}{Theorem}[section]
\newtheorem{corollary}[theorem]{Corollary}
\newtheorem{lemma}[theorem]{Lemma}

\theoremstyle{definition}

\newtheorem{remark}[theorem]{Remark}

\begin{document}


\begin{abstract}
We present a simple method,  not based on the transfer
matrices, to prove vanishing of dynamical transport
exponents. The method is applied to long range quasiperiodic operators.

\end{abstract}
\maketitle

 \section{Introduction}

Jean Bourgain, partially with collaborators, has developed a powerful method to
prove Anderson localization for ergodic Schr\"odinger operators, see
\cite{bbook} and references therein. The method relies heavily, in both perturbative and nonperturbative
settings, on the subharmonic function theory and the
theory of semi-algebraic sets and has turned out to be quite robust. While the precursor was the
non-perturbative approach of \cite{j} that initiated the emphasis on
obtaining off-diagonal  Green's function decay using bulk features
rather than individual eigenfunctions, Bourgain's method has crystallized
and developed the key ideas that did not require transfer matrices/nearest-neighbor Laplacians, thus allowing, in particular, the extension to Toeplitz
matrices as well as multidimensional localization
results.\footnote{See also \cite{gafa,liu} for streamlining and simplification of Bourgain's
multidimensional method and the
non-self-adjoint version.}

Discrete quasiperiodic operators with the Laplacian replaced by a
Toeplitz operator appear naturally in the context of Aubry duality,
and have been studied by several authors. Let
$H_{\theta,\alpha,\epsilon},$ with
$(\theta,\alpha)\in \T^2,$
act on $\ell^2(\Z)$ by 

\begin{equation}\label{bop}
(H_{\theta,\alpha,\epsilon}u)_n :=\epsilon \left( \sum_{k\in \Z}a_{n-k}u_{k}\right)+ v (\theta +n\alpha) u_n.
\end{equation}
where  $|a_n|\leq A_1 e^{-a|n|}$  for some $a,A_1>0$ and
$a_{-n}=\overline{a_{n}}$. Bourgain's main localization result for
the long-range case is
\begin{theorem}[\cite{bbook}, Theorem 11.20]\label{btheorem}
If $v$ is analytic non-constant on $\T$, then for
$|\epsilon|\leq \epsilon_0,\;\epsilon_0=\epsilon_0(A_1, a,v),$ $H_{\theta,\alpha,\epsilon}$ satisfies
Anderson localization for a full measure set of $(\theta,\alpha)\in
\T^2.$
\end{theorem}
We note that this theorem is {\it non-perturbative}, that is
$\epsilon_0$ does not depend on $\alpha.$ There is also a stronger, {\it arithmetic} (that is with an arithmetic
full measure condition on the frequency and phase) localization result for
$v(\theta)=\cos 2\pi \theta$ \cite{bj}, and recently an arithmetic multidimensional
result was obtained as a corollary of dual quantitative reducibility in
\cite{gy}, but for general function $v$ Bourgain's non-arithmetic theorem \ref{btheorem}
remains the strongest available. We note that the perturbative multidimensional
version appears in \cite{gafa}; however, in the multidimensional
case  there is no essential difference between the nearest neighbor and
long-range Laplacians.

At the same time, Anderson localization (pure point spectrum with
exponentially decaying eigenfunctions) is extremely fragile.  Indeed, it was shown by
Gordon \cite{Gor} and del Rio, Makarov and Simon \cite{DMS} that a
generic rank one perturbation of an operator with an interval in the spectrum
even in the regime of {\it dynamical} localization leads to singular continuous
spectrum, and therefore, by the RAGE theorem, growth of the moments. However, it was
shown in \cite{DJLS} that, under the condition of SULE, present in
many models, this growth can be at most logarithmic, and
thus preserves vanishing of the dynamical exponents.  Thus one can
argue that it is vanishing of the dynamical exponents $\beta (q)$ (see
(\ref{dyn-exp}) for the definition) that captures the physically relevant
effect of localization. 

Indeed, such localization-type results (vanishing of $\beta (q)$) have
been obtained, in increasing generality, for random and quasiperiodic operators
as a corollary of positive Lyapunov exponents in \cite{7,8,jsb, 15}, with \cite{rui}
covering the entire class of ergodic operators with base dynamics of zero topological
entropy, a class that includes shifts and skew-shifts on
higher-dimensional tori. 
Clearly, those techniques are transfer-matrix
based, thus don't extend to long-range operators.

In this note we present a very simple method to obtain such quantum dynamical
upper bounds for the long-range case, and show that one part
of Bourgain's localization proof can serve as an input to obtain an
arithmetic result: vanishing of
quantum dynamical exponents for all long-range quasiperiodic operators
with Diophantine frequencies, all phases, and sufficiently large analytic
potentials, see Corollary \ref{thmap1}. 
This should be contrasted with the
non-arithmetic result Theorem \ref{btheorem}. We note also that Anderson localization for
{\it all} phases does not even hold \cite{js,jl2}

Bourgain's method consists of multiple parts, and the one in question is establishing the sublinear
bound (\ref{sublinear})  for the number of boxes of size $N^c$ in a
box of size $N,$ that don't have the off-diagonal Green
function decay. Our method requires only presence
of {\it one} box of size $N^c$  {\it with} the off-diagonal Green
function decay, in a box of size $N,$ thus Bourgain's sublinear bound is even an
overkill for a needed input.

We note that, while suitable for long-range, our method is still
one-dimensional, as only in dimension one does one box create a
barrier and thus a good estimate for the Green's
function in a bigger box. Yet it does provide the first departure from the Lyapunov
exponent/transfer matrix based methods, and leads to a strong
corollary. Also, it extends easily to the (not necessarily
uniform) band, requiring only  one ``good box'' to apply Theorem \ref{mainthm2}.

Let us now introduce the main concepts. We restrict here to dimension
one, although many of the statements and definitions are easily
extendable to higher dimensions.  For a fixed self-adjoint operator
$H$ on $\ell^2(\Z), $ $\phi\in \ell^2(\Z)$ and $p, T > 0$,  let
\begin{equation}\label{mom}
    \langle|X|_{\phi}^p\rangle(T)=\frac{2}{T}\int_0^\infty e^{-2t/T}\sum_{n\in\Z}|n|^p|(e^{- itH}\phi,\delta_n)|^2
\end{equation}
The growth rate of $\langle|X|_{\phi}^p\rangle(T)$ characterizes how
fast does $e^{- itH}\phi$ spread out. The power law bounds for
$\langle|X|_{\phi}^p\rangle(T)$ are naturally characterized by the
following upper transport  exponents  $\beta_{\phi}^+(p)$ defined  as
\begin{equation}\label{dyn-exp}
    \beta_{\phi}^+(p)=\limsup_{T\to\infty}\frac{\ln\langle|X|_{\phi}^p\rangle(T)}{p\ln T}.
\end{equation}

Here we study 
Schr\"odinger operators on $\ell^2(\Z)$ of the form,
\begin{equation*}
H=A+V,
\end{equation*}
where  $V=\{V_n\}_{n\in\Z}$ is  real bounded  and $A$ is a  long-range  operator of the form
\begin{equation*}
(Au)_n=\sum _{k\in\Z} a_{n-k} u_k,
\end{equation*}
where $|a_n|\leq A_1 e^{-a|n|}$  for some $a,A_1>0$ and $a_{-n}=\overline{a_{n}}$.

More precisely, 
\begin{equation}\label{g01}
(Hu)_n = \left( \sum_{k\in \Z}a_{k}u_{n-k}\right)+V_n u_n.
\end{equation}
Just like Schr\"odinger operators, such operators admit a ballistic
bound on the transport exponents
\begin{theorem}\label{mainthm1}
	Let $H$ be given by \eqref{g01}.
Assume $\phi$ is compactly supported.
Then the  upper  transport exponent
	$\beta_{\phi}^+(q)\leq1$ for any $q>0$.
\end{theorem}
\begin{remark}  In fact, sufficiently fast decay
works equally well, but we restrict in all results, to the compactly supported $\phi$, for simplicity.
\end{remark}
Theorem \ref{mainthm1} is probably well known, but we didn't find the
proof in the literature. The proof,  following the ideas of
\cite{rs78,sim90}, is presented in the appendix.

Let  $R_{{\Lambda}}$  be the operator of  restriction to $\Lambda \subset \Z$. Define the Green's function by
\begin{equation}\label{g0}
G_{{\Lambda}}(z)=(R_{{\Lambda}}(H-zI)R_{{\Lambda}})^{-1}.
\end{equation}
Set $ G(z)=(H-zI)^{-1}$.  Clearly, both $  G_{{\Lambda}}(z)$ and $G(z)$ are always well defined for $z\in \C_+\equiv \{z\in \C: \Im z>0\}$. Sometimes, we drop the dependence on $z$  for simplicity.
Since the operator $H$ given by \eqref{g01} is bounded, there  exists
$K>0$  such that $\sigma({H})\subset [-K+1,K-1]$. Our main general
result is
\begin{theorem}\label{mainthm2}
		Let $H$ be given by \eqref{g01}.
	Suppose there exist    $\delta>0$ and $N_0>0$ such that the following is true.
	Let  $z=E+i\varepsilon$ with $|E|\leq K$ and $0<\varepsilon\leq \delta$.
	Suppose for    $N>N_0$, there exists an  interval $I\subset [-\frac{ N}{2}, -\frac{ N}{4}]$ or $I\subset [\frac{ N}{4}, \frac{ N}{2}]$ such that
	$|I|  \geq N^{\delta}  $
	and  for any $n,n^\prime\in {I}$ and $|n-n^\prime|\geq \frac{1}{20}|I|$,  we have
	\begin{equation*}
	| G_{{I}}(z)(n;n^\prime)|\leq e^{-|I|^{\delta}}.
	\end{equation*}
	Assume $\phi$ is compactly supported. 
	Then the upper transport exponent $\beta_{\phi}^+(p)=0 $ for any  $p>0$.
\end{theorem}
\begin{remark} For the Schr\"odinger case, the existence of such interval $I$ (in fact,
a stronger statement, but this is not important) can be deduced from
the positive Lyapunov exponents and Cramer's rule  by the method going
back to \cite{jlast3}.\end{remark}

We say   $\alpha\in\R$ 
 is Diophantine if there exist $\kappa$ and $ \tau>0$
such that for any  $k\in \Z\backslash\{0\}$,
\begin{equation*}
||k\alpha||_{\R/\Z} \geq \frac{\tau}{|k|^{\kappa}},
\end{equation*}
where $||x||_{\R/\Z}=\text{dist}(x,\Z)$.


Let $ H_{\alpha,\theta,\epsilon} $ be as in (\ref{bop}). Fixing
$\alpha$ and $\epsilon$,  we denote the
$\beta_{\phi}^+(p)$ for operator $ H_{\alpha,\theta,\epsilon} $ by $\beta_{\phi,\theta}^+(p).$ Our main application is
\begin{corollary}\label{thmap1}
There exists an $\epsilon_0=\epsilon_0(v,A_1,a)>0$ such that for any
compactly supported $\phi$ and Diophantine $\alpha,$  
$\beta_{\phi,\theta}^+(p)=0 $
for any $|\epsilon|\leq \epsilon_0,$ any $\theta \in \R$ and $p>0$.
\end{corollary}
It immediately implies also
\begin{corollary}\label{cor2}

There exists an $\epsilon_0=\epsilon_0(v,A_1,a)>0$ such that  for any $\phi\in \ell^2(\Z)$
the spectral measure $\mu_{\phi}$ of operator
$H_{\theta,\alpha,\epsilon}$ is zero dimensional for any $\theta \in \R,$  Diophantine $\alpha,$
and any $|\epsilon|\leq \epsilon_0$.
\end{corollary}

 \section{Proof of Theorem \ref{mainthm2}}

For the Schr\"odinger case, the proof would be just a double application of
the resolvent identity:
\begin{align*}
 G=&G_I +G_{I^c}-(G_I +G_{I^c}) (H-H_I-H_{I^c})(G_I +G_{I^c})\\
 &+(G_I +G_{I^c}) (H-H_I-H_{I^c}) (G_I +G_{I^c}) (H-H_I-H_{I^c})G
\end{align*}
ensuring the decay of $|G(0,n)|$ based on the ``barrier'' box $I.$
The problem with the long-range case is that such expansion for
$G(0,N)$ will contain terms all grouped nearby, thus neither incorporating
the decay coming from the barrier box nor from $|a_n|.$ In order to
tackle this difficulty we introduce several extra steps, all involving
applications of the
resolvent identity, but with different boxes.\\

Since $\phi$  has a compact support, there exists $K_1$ such that 
$\phi(n)=0$ for $|n|\geq K_1$.

   Assume $T>\frac{1}{\delta}$.
   Fix $z=E+i\frac{1}{T}$ with $|E|\leq K$.  Below, $C$ ($c$) is  a large (small) constant that may depend  on $\delta$,  $K$,  $A_1$, $a$, $\phi$ and $V=\{V_n\}$.
   Let $I=[b-\ell,b+\ell]$ with $\ell> 0 $ and
   $b$ such that $|b\pm\ell|$ is large. 
   Suppose
   \begin{equation}\label{smallgreen}
   |G_{I}(m,n)|\leq C e^{-c\ell^c}
   \end{equation}
   for any $m\in I,n\in {I}$ and $|m-n|\geq \frac{1}{20}\ell$.
   
   Recall that if
   \begin{equation*}
   \Lambda=\Lambda_1\cup\Lambda_2, \Lambda_1\cap\Lambda_2=\emptyset,
   \end{equation*}
   then
   \begin{equation*}
   G_{{\Lambda}}=G_{{\Lambda}_1}+G_{{\Lambda}_2} -(G_{{\Lambda}_1}+G_{{\Lambda}_2})(H_{{\Lambda}} -H_{{\Lambda}_1}-H_{{\Lambda}_2}) G_{{\Lambda}},
   \end{equation*}
   (provided the relevant matrices $R_{{\Lambda}}(H-zI)R_{{\Lambda}}$ and $R_{{\Lambda}_i}(H-zI)R_{{\Lambda}_i}$ are invertible)
   where $H_{{\Lambda}}=R_{{\Lambda}}HR_{{\Lambda}}$.
   
   If $m\in {\Lambda}_1$ and $n\in {\Lambda}$, we have
  \begin{equation}\label{Bolckgreen}
    G_{\Lambda}(m,n)=G_{\Lambda_1}(m,n)\chi_{\Lambda_1}(n)-\sum_{n_1\in \Lambda_1,n_2\in\Lambda_2} G_{\Lambda_1}(m,n_1) a_{n_1-n_2} G_{\Lambda}(n_2,n).
  \end{equation}
  Therefore,
   \begin{equation}\label{g14}
 | G_{\Lambda}(m,n)|\leq |G_{\Lambda_1}(m,n)\chi_{\Lambda_1}(n)|+C\sum_{n_1\in \Lambda_1,n_2\in\Lambda_2} |G_{\Lambda_1}(m,n_1)|e^{-c|n_1-n_2|} | G_{\Lambda}(n_2,n)|.
  \end{equation}
  \begin{lemma}\label{Keyle1}
  Assume that  for some interval $I=[b-\ell,b+\ell]$ and $z=E+\frac{i}{T}$, (\ref{smallgreen}) holds.
  Then
  \begin{equation}\label{smallgreen1}
    |G_{\Lambda}(m,n)|\leq CT^2e^{-c\ell^c}
 \end{equation}
   for any $n\in I $, $m\in[b-\ell+\frac{\ell}{10}, b+\ell]$ and $|m-n|\geq \frac{1}{10}\ell$, where $\Lambda=(-\infty,b+\ell]$.
  \end{lemma}
   \begin{proof}
  	Let $\Lambda_1=I=[b-\ell,b+\ell]$ and $\Lambda_2=(-\infty,b-\ell-1]$. Clearly, $\Lambda=\Lambda_1\cup \Lambda_2$.
  	By \eqref{smallgreen} and  (\ref{g14}), one has that
  	\begin{equation}\label{Bolckgreen1}
    | G_{\Lambda}(m,n)|\leq  Ce^{-c \ell ^c}+C\sum_{n_1\in \Lambda_1,n_2\in\Lambda_2} |G_{\Lambda_1}(m,n_1)|e^{-c|n_1-n_2|} | G_{\Lambda}(n_2,n)|.
  	\end{equation}
  	It suffices to bound the second  term on the right of  (\ref{Bolckgreen1}).
  	
  	For any $n_1\in\Lambda_1$,
  	\begin{equation}\label{g15}
  	\sum_{n_2\in\Lambda_2} e^{-c|n_1-n_2|}\leq C.
  	\end{equation}
    If $n_1\in[b-\ell, b-\ell+\frac{\ell}{20}]$, by the fact that $m\in[b-\ell+\frac{\ell}{10}, b+\ell]$ and (\ref{smallgreen}),
  one has
  \begin{equation}\label{equ1}
  | G_{\Lambda_1}(m,n_1)|\leq Ce^{-c\ell^c}.
  \end{equation}
  
  If $n_1\in[ b-\ell+\frac{\ell}{20},b+\ell]$, 
  one has
  \begin{equation}\label{equ2}
 \sum_{n_2\in \Lambda_2} e^{-c|n_1-n_2|}\leq Ce^{-c\ell}.
  \end{equation}

 Since $\Im z=\frac{1}{T}$, one has  that 
  \begin{equation}\label{equ3}
  | G_{\Lambda_1}(m,n_1)|\leq T,|G_{\Lambda}(n_2,n)|\leq     T.
  \end{equation}
  
 By \eqref{g15},  (\ref{equ1}),(\ref{equ2}) and (\ref{equ3}),
  we have
  \begin{equation*}
\sum_{n_1\in \Lambda_1,n_2\in\Lambda_2} |G_{\Lambda_1}(m,n_1)|e^{-c|n_1-n_2|} | G_{\Lambda}(n_2,n)|\leq C T^2e^{-c\ell^c}
  \end{equation*}
  \end{proof}
  \begin{lemma}\label{Keyle2}
  Assume $b-\ell$ is large. Under the conditions of Lemma \ref{Keyle1}, we have  that for any $j$ with $|j|\leq K_1$ and  $n\in[b+\ell-\frac{\ell}{10}, b+\ell]$, 
  \begin{equation}\label{smallgreen2}
    |G_{\Lambda}(j,n)|\leq CT^4e^{-c\ell^c},
 \end{equation}
   where $\Lambda=(-\infty,b+\ell]$.
  \end{lemma}
  \begin{proof}
   Let $\Lambda_2=[b-\ell,b+\ell]$, $\Lambda_1=(-\infty,b-\ell-1]$ and $\Lambda=(-\infty,b+\ell]$.
  By (\ref{g14}), one has that for any  $j$ with $|j|\leq K_1$, 
  \begin{equation}\label{Bolckgreen2}
   | G_{\Lambda}(j,n)|\leq C\sum_{n_1\in \Lambda_1,n_2\in\Lambda_2} |G_{\Lambda_1}(j,n_1) |e^{-c|n_1-n_2|} |G_{\Lambda}(n_2,n)|.
  \end{equation}

	For any $n_2\in\Lambda_2$,
\begin{equation}\label{g16}
\sum_{n_1\in\Lambda_1} e^{-c|n_1-n_2|}\leq C.
\end{equation}

  If $n_2\in[b-\ell, b+\ell-\frac{\ell}{5}]$, by the fact that $n\in[b+\ell-\frac{\ell}{10}, b+\ell]$ and (\ref{smallgreen1}),
  one has
  \begin{equation}\label{equ5}
   |G_{\Lambda}(n_2,n)|= |G_{\Lambda}(n,n_2)|\leq CT^2e^{-c\ell^c}.
  \end{equation}

  If $n_2\in[ b+\ell-\frac{\ell}{5},b+\ell]$, by the fact that  $n_1\leq b-\ell$,
  one has
  \begin{equation}\label{equ6}
 \sum_{n_1\in\Lambda_1} e^{-c|n_1-n_2|}\leq Ce^{-c\ell}.
  \end{equation}

 By \eqref{g16},  (\ref{equ5}),(\ref{equ6}) and (\ref{equ3}),
  we have
  \begin{equation*}
  \sum_{n_1\in \Lambda_1,n_2\in\Lambda_2} |G_{\Lambda_1}(j,n_1)|e^{-c|n_1-n_2|} |G_{\Lambda}(n_2,n)|\leq CT^4e^{-c\ell^c}.
  \end{equation*}
 This implies  (\ref{smallgreen2}).
  \end{proof}
  \begin{lemma}\label{Keyle3}
  	Let $z=E+\frac{i}{T}$. 
 Assume that $\ell\geq |\tilde{N}|^{\delta}$ and for some  interval $I=[b-\ell,b+\ell]$ with $I\subset \left[\frac{|\tilde{N}|}{4}, \frac{|\tilde{N}|}{2}\right]$ or  $I\subset \left[-\frac{|\tilde{N}|}{2}, -\frac{|\tilde{N}|}{4}\right]$,  (\ref{smallgreen}) holds.
  Then for any $j$ with $|j|\leq K_1$, 
  \begin{equation}\label{smallgreen3}
    |G_{\Lambda}(j,\tilde{N})|\leq CT^6e^{-c{|\tilde{N}}|^c},
 \end{equation}
  where $\Lambda=(-\infty,\infty)$.
  \end{lemma}
  \begin{proof}
  	Without loss of generality, assume $\tilde{N}>0$.
   Let $\Lambda_1=(-\infty,b+\ell]$, $\Lambda_2=[b+\ell+1,\infty)$ and $\Lambda=(-\infty,\infty)$.
  By (\ref{g14}), one has
  \begin{equation}\label{Bolckgreen3}
   | G_{\Lambda}(j,\tilde{N})|\leq C\sum_{n_1\in \Lambda_1,n_2\in\Lambda_2} |G_{\Lambda_1}(j,n_1)| e^{-c|n_1-n_2|} |G_{\Lambda}(n_2,N)|.
  \end{equation}

First, one has
\begin{equation}\label{g17}
\sum_{n_1\in \Lambda_1,n_2\in\Lambda_2} e^{-c|n_1-n_2|}\leq C.
\end{equation}
  If $n_1\in[b+\ell-\frac{\ell}{10}, b+\ell]$,  By  (\ref{smallgreen2}),
  one has
  \begin{equation}\label{equ7}
   |G_{\Lambda_1}(j,n_1)|\leq CT^4e^{-c\ell^c}.
  \end{equation}

  If $n_1\in( -\infty,b+\ell-\frac{\ell}{10}]$ and  $n_2\in \Lambda_2$,
  one has
  \begin{equation}\label{equ8}
\  e^{-c|n_1-n_2|}\leq Ce^{-c\ell}.
  \end{equation}

 By \eqref{g17},  (\ref{equ7}), (\ref{equ8}) and  (\ref{equ3}),
  we have
  \begin{equation*}
 \sum_{n_1\in \Lambda_1,n_2\in\Lambda_2} |G_{\Lambda_1}(j,n_1)| e^{-c|n_1-n_2|} |G_{\Lambda}(n_2,\tilde{N})| \leq CT^6e^{-c\ell^c}.
  \end{equation*}
 This implies  (\ref{smallgreen3}).
  \end{proof}

  \begin{proof}[\bf Proof of Theorem \ref{mainthm2}]
  	This is standard. For any $j$ with $|j|\leq K_1$,
	let
	\begin{equation}\label{g-1}
	a(j,n,T)=\frac{2}{T}\int_{0}^{\infty}e^{-2t/T}|(e^{- itH}\delta_j,\delta_n)|^2 dt.
	\end{equation}
	 By the Parseval formula
	\begin{equation}\label{Par}
	a(j,n,T)=\frac{1}{T\pi} \int_{-\infty}^{\infty}|((H-E-\frac{i}{T})^{-1}\delta_j,\delta_n)|^2 d E.
	\end{equation}
	Recall that $\sigma(H)\subset [-K+1,K-1]$. 
For any $E\in(-\infty,-K)\cup  (K,\infty)$,   $\eta=\text{dist} (E+\frac{i}{T},\text{spec}(H))\geq 1$.
The  well-known Combes-Thomas estimate yields for large $n$,
	\begin{equation}\label{CT}
	|((H-E-\frac{i}{T})^{-1}\delta_j,\delta_n)|\leq Ce^{-c|n|}.
	\end{equation}
By \eqref{Par} and \eqref{CT},  one has that 
	\begin{equation}\label{Par1}
	a(j,n,T)\leq Ce^{-c|n|}+\frac{1}{T\pi} \int_{-K}^{K}|((H-E-\frac{i}{T})^{-1}\delta_j,\delta_n)|^2 d E.
	\end{equation}
	By Lemma \ref{Keyle3},
	we have for any $|E|\leq K$,
	\begin{equation}\label{g-2}
	|((H-E-\frac{i}{T})^{-1}\delta_j,\delta_n)|\leq CT^6 e^{-c\ell^c}\leq CT^6 e^{-c|n|^c}.
	\end{equation}
	By \eqref{Par1} and  \eqref{g-2},  one has that 
	\begin{equation}\label{ant}
	a(j,n,T)\leq CT^{11}e^{-c|n|^c }.
	\end{equation}
	Therefore,
	
	 \begin{eqnarray}
	\langle|X|_{\phi}^p\rangle(T) &\leq & C \sum_{|j|\leq K_1}\sum_{n\in\Z}|n|^pa(j,n,T) \nonumber\\
	&\leq & \sum_{n\in\Z}CT^{11}|n|^p e^{-c|n|^c }\nonumber\\
	&\leq & CT^{11} .
	\end{eqnarray}
	It  implies
	\begin{equation*}
	\beta_{\phi}^+(p)\leq \frac{11}{p}.
	\end{equation*}
	Since   $\beta_{\phi}^+(p)$ are nondecreasing, we have that for every $p>0$,
	\begin{equation}\label{etau1}
	\beta_{\phi}^+(p)\leq \lim_{p\to \infty}\beta_{\phi}^+(p) =0.
	\end{equation}
	
\end{proof}
 \section{Proof of Corollary \ref{thmap1} }
Under the assumption of Corollary  \ref{thmap1}, one has that when
$|\epsilon| \leq \epsilon_0$,  the following  
holds for some $\delta_0>0$,
	\begin{equation}\label{sublinear}
\#\{b\in\Z: |b|\leq N, I_b\text{ does not satisfy } \eqref{smallgreen} \}\leq N^{1-\delta_0}.
	\end{equation} 
This was proved
in Ch. 11 of \cite{bbook} for $E\in\R$ and also holds for complex energies
\cite{liu}\footnote{This is mentioned already in \cite{bbook}}. Therefore,  Corollary \ref{thmap1}  follows from Theorem \ref{mainthm2}.
\qed

\appendix \section{Proof of  Theorem \ref{mainthm1}}
\begin{proof}
By an easy application of H\"older inequality, $\beta_{\phi}^+(q)$ is nondecreasing with respect to $q$.  Therefore, it suffices  to
	show that for any $N\in \mathbb{N}$, $\beta_{\phi}^+(2N)\leq 1$.
	
	Define the free
	long range Schr\"odinger operator by
	\begin{equation*}
	(H_0u)_n =  \sum_{k\in \Z}a_{n-k}u_{k}=\sum_{k\in\Z} a_ku_{n-k}.
	\end{equation*}
	For any sequence  $\gamma=\{\gamma_k\}$ with $|\gamma_k|\leq Ce^{-c|k|}$, we  define the momentum operator $X_{2p}^{\gamma}$:
	\begin{equation*}
	(X_{2p}^\gamma u)_n= n^{p}\sum_k \gamma_ku_{n-k},
	\end{equation*}
	and  $\hat{X}_{2p}^{\gamma}=-i[H_0,X_{2p}^\gamma]$, where $[B_1,B_2]=B_2B_1-B_1B_2$.

Direct computations  implies,
\begin{align}
	(\hat{X}_{2p}^{\gamma}u)_n  &=-i \left(n^p \sum_{k_1\in\Z} \gamma_{k_1} (H_0u)_{n-k_1} -\sum_{k\in\Z} a_k (X_{2p}^\gamma u)_{n-k}\right)\nonumber\\
	&=-i \left(n^p \sum_{k_1\in\Z,k\in\Z} \gamma_{k_1}  a_ku_{n-k_1-k} -\sum_{k\in\Z,k_1\in\Z} a_k (n-k)^p\gamma_{k_1} u_{n-k-k_1}\right)\nonumber\\
	&=-i\sum_{k\in\Z,k_1\in\Z}(n^p-(n-k)^p)a_k\gamma_{k_1}u_{n-k-k_1}\nonumber\\
		&=-i\sum_{m\in \Z} \left(\sum_{k\in\Z}(n^p-(n-k)^p)a_k\gamma_{m-k}\right) u_{n-m}\nonumber
\end{align}
Therefore,    $\hat{X}_{2p}^{\gamma} $  can be rewritten as
	\begin{equation}\label{g4}
	\hat{X}_{2p}^{\gamma}= \sum_{j=0}^{p-1} {X}_{2j}^{\gamma^j},
	\end{equation}
	for some new sequences  $\{\gamma^j_k\}$ with  $|\gamma^j_k|\leq C_je^{-c_j|k|}$, $j=0,1,\cdots,p-1$.

	Let
	\begin{equation*}
	X_{2p}^{\gamma}(t)=e^{itH}  X_{2p}^{\gamma} e^{-itH}, \hat{X}_{2p}^\gamma(t)=e^{itH}  \hat{X}_{2p}^\gamma e^{-itH}.
	\end{equation*}
	
	Differentiating  $ X_{2p}^{\gamma}(t)$,  one has that 
	\begin{equation}\label{g7}
	\frac{dX_{2p}^{\gamma}(t)}{dt}=\hat{X}_{2p}^\gamma(t).
	\end{equation}
	
	We will show inductively that
	\begin{equation}\label{g1}
	(X_{2N}^\gamma (t)\phi,X_{2N}^\gamma(t)\phi)\leq C_{\gamma,\phi,N}t^{2N} \text{ for large } t.
	\end{equation}

	We first prove (\ref{g1}) for $N=1$.
	Differentiating $ X_{2}^\gamma(t)$,  one has that
	\begin{equation}\label{g13}
	\frac{dX_{2}^\gamma(t)}{dt}=\hat{X}_{2}^\gamma(t),
	\end{equation}
	where $\hat{X}_{2}^\gamma(t) $ is a bounded selfadjoint operator by (\ref{g4}).
By \eqref{g7}, one has
	\begin{equation}\label{g5}
	X_{2}^\gamma (t)= X_{2}^\gamma+\int_{0}^t \hat{X}_{2}^\gamma(s)ds.
	\end{equation}
	
	This implies
	$$	(X_{2}^\gamma(t)\phi,X_{2}^\gamma(t)\phi) \;\;\;\;\;\;\;\;\;\; \;\;\;\;\;\;\;\;\;\; \;\;\;\;\;\;\;\;\;\; \;\;\;\;\;\;\;\;\;\; \;\;\;\;\;\;\;\;\;\; \;\;\;\;\;\;\;\;\;\; \;\;\;\;\;\;\;\;\;\; \;\;\;\;\;\;\;\;\;\; \;\;\;\;\;\;\;\;\;\;$$
	\begin{eqnarray*}
	&=& (X_{2}^\gamma\phi+\int_{0}^t \hat{X}_{2}^\gamma(s)\phi ds, X_{2}^\gamma\phi+\int_{0}^t \hat{X}_{2}^\gamma(s)\phi ds) \\
		&\leq &  ||X_{2}^\gamma\phi||^2+ 2||X_{2}^\gamma\phi||\;\;\;  \int_{0}^t|| \hat{X}_{2}^\gamma(s)\phi ||ds + \left(\int_{0}^t ||\hat{X}_{2}^\gamma(s)\phi ||ds\right)^2\\
		&\leq&  C_{\gamma,\phi} t^2+C_{\gamma,\phi}  t+C_{\gamma,\phi} ,
	\end{eqnarray*}
	since $\phi$ has compact support and $\hat{X}_{2}^\gamma(t) $ is  bounded.
	
	Assume that  (\ref{equ1}) holds  for $p\leq N-1$. This means
        that  for any sequence $\{\gamma_k\}$ and $p=1,2,\cdots N-1$,
	\begin{equation}\label{g6}
	({X}_{2p}^{\gamma}(t)\phi,{X}_{2p}^\gamma(t)\phi)\leq C_{\gamma,\phi,p}t^{2p} \text{ for large } t.
	\end{equation}

	By (\ref{g7}), one has
	\begin{equation}\label{g8}
	X_{2N}^{\gamma}(t)= X_{2N}^{\gamma}+\int_{0}^t \hat{X}_{2N}^\gamma(s)ds.
	\end{equation}
	
	By (\ref{g4}) and (\ref{g6}),
	we have
	\begin{equation}\label{g10}
	||\hat{X}_{2N}^\gamma(t)\phi||\leq C_{\gamma,\phi,N} t^{N-1} \text{ for large } t.
	\end{equation}
	This implies, for large $t$,
	$$	(X_{2N}^\gamma(t)\phi,X_{2N}^\gamma(t)\phi) \;\;\;\;\;\;\;\;\;\; \;\;\;\;\;\;\;\;\;\; \;\;\;\;\;\;\;\;\;\; \;\;\;\;\;\;\;\;\;\; \;\;\;\;\;\;\;\;\;\; \;\;\;\;\;\;\;\;\;\; \;\;\;\;\;\;\;\;\;\; \;\;\;\;\;\;\;\;\;\; \;\;\;\;\;\;\;\;\;\;$$
	\begin{eqnarray*}
	&=& \left(X_{2N}^\gamma\phi+\int_{0}^t \hat{X}_{2N}^\gamma(s)\phi ds, X_{2N}^\gamma\phi+\int_{0}^t \hat{X}_{2N}^\gamma(s)\phi ds\right) \\
		&\leq &  ||X_{2N}^\gamma\phi||^2+ 2||X_{2N}^\gamma\phi||\;\;\;  \int_{0}^t|| \hat{X}_{2N}^\gamma(s)\phi ||ds + \left(\int_{0}^t ||\hat{X}_{2N}^\gamma(s)\phi ||ds\right)^2\\
		&\leq&  C_{\gamma,\phi,N} t^{2N}.
	\end{eqnarray*}
	
	Let  $\{\gamma_k\}$ be the sequence such that $\gamma_0=1$ and $\gamma_k=0$ for $k\neq 0$.
	Therefore, one has
	\begin{eqnarray}
		\sum_{n\in\Z} |n|^{2N}|(e^{-itH}\phi,\delta_n)|^2&=&  (X_{2N}^\gamma e^{-itH}\phi,X_{2N}^\gamma e^{-itH}\phi) \nonumber\\
		&=& (e^{itH}X_{2N}^\gamma e^{-itH}\phi,e^{itH}X_{2N}^\gamma e^{-itH}\phi)\nonumber \\
		&=& (X_{2N}^\gamma(t)^\gamma\phi,X_{2N}^\gamma(t)\phi) \nonumber\\
		&\leq &C_{\phi,N} t^{2N}.\label{g11}
	\end{eqnarray}
	
	By (\ref{mom}) and \eqref{g11}, one has
	\begin{eqnarray*}
		\langle|\hat{X}|_{\phi}^{2N}\rangle(T) &\leq & \frac{2}{T}\int_0^\infty e^{-2t/T}C_{\phi,N} t^{2N}dt \\
		&\leq &  C_{\phi,N} T^{2N}.
	\end{eqnarray*}
	
	Thus    $\beta_{\phi}^+(q)\leq1$ for any $q>0$.
\end{proof}
 \section*{Acknowledgments}
This paper is devoted to the memory of Jean Bourgain. 
Both authors have been profoundly influenced by Jean. S.J. was
fortunate to experience direct influence, see
\cite{notices}. W.L. believes he became a
mathematician through detailed  reading of \cite{bbook}. \\ We are grateful to the Isaac Newton Institute for Mathematical Sciences,
Cambridge, for its hospitality, supported by EPSRC Grant Number EP/K032208/1, during
the programme Periodic and Ergodic Spectral Problems where this work
was started. S.J. was supported by a Simons Foundation Fellowship.
Her work was also supported by  NSF DMS-1901462 and  DMS-2052899. W.L. was supported by the 
NSF DMS-2000345 and DMS-2052572.

\end{document}